\documentclass[11pt,a4paper,english]{amsart}
\usepackage[utf8]{inputenc}

\usepackage{amsmath,amsfonts,amsthm}
\usepackage{url}
\usepackage{esint}
\usepackage{fullpage}
\usepackage{graphicx}
\newtheorem{theorem}{Theorem}[section]

\newtheorem{remark}[theorem]{Remark}
\newtheorem{lemma}[theorem]{Lemma}

\DeclareMathOperator*{\Tr}{Tr}
\let\Re\relax
\let\Im\relax
\DeclareMathOperator*{\Re}{Re}
\DeclareMathOperator*{\Im}{Im}

\newcommand{\Trper}{\underline{\Tr}}

\newcommand{\per}{\mathrm{per}}
\newcommand{\defect}{\mathrm{def}}
\newcommand{\nucl}{\mathrm{nucl}}
\newcommand{\vc}{{v_{\rm c}}}
\newcommand{\vck}{{v_{\rm c,k}}}
\newcommand{\vcper}{{v_{\per}}}

\newcommand{\R}{\mathbb R}
\newcommand{\C}{\mathbb C}
\newcommand{\N}{\mathbb N}
\newcommand{\Z}{\mathbb Z}

\newcommand{\cH}{\mathcal H}

\newcommand{\cR}{\mathcal R}

\newcommand{\cN}{\mathcal N}

\newcommand{\BZ}{\mathcal B}
\newcommand{\Schatten}{{\mathfrak S}}
\newcommand{\cont}{\mathcal C}
\newcommand{\kerker}{\mathcal K}

\begin{document}

\title{Screening in the finite-temperature reduced Hartree-Fock model}
\author{Antoine Levitt}
\begin{abstract}
  We prove the existence of solutions of the reduced Hartree-Fock
  equations at finite temperature for a periodic crystal with a small
  defect, and show total screening of the defect charge by the
  electrons. We also show the convergence of the damped
  self-consistent field iteration using Kerker preconditioning to
  remove charge sloshing. As a crucial step of the proof, we define
  and study the properties of the dielectric operator.
\end{abstract}
\maketitle
\section{Introduction}
A point charge $Q$ placed in vacuum creates an electric potential
$\frac{Q}{4\pi r}$, $r$ being the distance to the charge (in units where the permittivity of the
vacuum $\varepsilon_{0}$ is taken to be $1$). By contrast, when an
defect is placed in a material, the material reorganises itself: a
positive charge creates an energetically favorable region for the
electrons, which flock towards the defect. At equilibrium, they form a
``shield'' of negative charge, effectively screening the Coulomb
interaction at long range.

Phenomenologically, insulators and metals exhibit a different
screening behavior. In insulators, electrons are tightly bound to the
nuclei, and cannot deviate too much from their equilibrium position to
move towards the defect. Accordingly, the long-range behavior of the
total potential, including the effects of the electrons, is
$Q/(4\pi \varepsilon r)$, where $\varepsilon > 1$ is the dielectric
constant of the material. Thus, effectively, the charge $Q$ is scaled
by the dielectric constant $\varepsilon$: this is called partial
screening.

In metals, however, electrons are free to move in response to the
defect and totally screen it, so that the total potential becomes
effectively short-range. A simple model for the total potential is the Yukawa potential
\begin{align*}
  V(x) = \frac{Q e^{-k |x|}}{|x|}
\end{align*}
where $1/k$ is the screening length. At low temperatures however, $V$
displays an oscillatory behavior with a power-law decay, called
Friedel oscillations.

The purpose of this paper is to justify the total screening of small
defects in the reduced Hartree-Fock (rHF) model at finite temperature. This is to be
contrasted with the partial screening of insulators at zero
temperature obtained in \cite{cances2010dielectric} in the same model:
at finite temperature, electrons are mobile and behave as in a metal.
We also justify the Kerker preconditioning scheme, which neutralizes
the ``charge sloshing'' effect that slows down simple self-consistent
iterations in extended systems \cite{kerker1981efficient}.

For a finite system of $N_{\rm el}$ electrons in an external potential $V_{\rm ext}$, the
reduced Hartree-Fock (rHF) equation for the total potential $V$ is given by
\begin{align*}
  \begin{cases}
    V = V_{\rm ext} + \vc F_{\varepsilon_F}(V),\\
    \int_{\R^{3}} F_{\varepsilon_F}(V) = N_{\rm el}.
  \end{cases}
\end{align*}
The Coulomb operator $\vc$ is given by the convolution
\begin{align*}
  (\vc \rho)(x) = \frac 1 {4\pi} \int_{\R^{3}} \frac{\rho(y)}{|x-y|} dy
\end{align*}
or in Fourier space
\begin{align*}
  \widehat{\vc \rho}(q) = \frac{\widehat \rho(q)}{|q|^{2}}.
\end{align*}
The potential-to-density mapping $F_{\varepsilon_F}$ is given by
\begin{align*}
  F_{\varepsilon_F}(V)(x) = f_{\varepsilon_{F}}(-\Delta + V)(x,x)
\end{align*}
where the Fermi-Dirac distribution $f_{\varepsilon_{F}}$ is
\begin{align*}
  f_{\varepsilon_{F}}(\varepsilon) = \frac{1}{1+ e^{\frac{\varepsilon-\varepsilon_{F}}{k_{B} T}}},
\end{align*}
with $T$ the temperature and $k_{B}$ the Boltzmann constant. The
density matrix $f_{\varepsilon_{F}}(H)$ is defined through the functional calculus of
self-adjoint operators, and $f_{\varepsilon_{F}}(H)(x,x)$ is the associated density (see
Section \ref{sec:def_density}). The Fermi level $\varepsilon_{F}$ is
determined through the charge neutrality condition $\int_{\R^{3}} F_{\varepsilon_F}(V) = N_{\rm
el}$.

This model, also called the Hartree model, random phase
approximation (RPA) or Schrödinger-Poisson, can be seen as a
simplification of Kohn-Sham density functional theory where the
exchange-correlation potential is neglected, or of the Hartree-Fock
model without the exchange term. In the zero-temperature case, it
derives from a convex variational principle, which allows for a
complete existence and uniqueness theory \cite{solovej1991proof}.

This convexity also means that it is possible to justify rigorously
the thermodynamic limit for periodic systems
\cite{catto2001thermodynamic}, something that seems out of reach for
the full Hartree-Fock or Kohn-Sham model. The resulting periodic model
takes the following form. Let $\cR$ be the crystal lattice, $\Gamma$ a
unit cell, and $W_{\nucl}$ the $\cR$-periodic potential created by the
nuclei. Then the periodic rHF model is
\begin{align}
  \label{eq:periodic_hartree}
  \begin{cases}
    W = W_{\nucl} + \vcper F_{\varepsilon_{F}}(W)\\
    \int_{\Gamma} F_{\varepsilon_{F}}(W) = N_{\rm el}
  \end{cases}
\end{align}
where $\vcper \rho$ is the unique periodic solution of
\begin{align}
  \label{eq:vcper}
  \begin{cases}
    -\Delta (\vcper \rho) = \rho - \frac 1{|\Gamma|}\int_{\Gamma} \rho\\
    \int_{\Gamma} (\vcper \rho) = 0
  \end{cases}
\end{align}
and $N_{\rm el}$ is now the number of electrons per unit cell. The
potential-to-density mapping takes the same form
$F_{\varepsilon_{F}}(W) = f_{\varepsilon_{F}}(-\Delta + W)(x,x)$, and
maps periodic potentials to periodic densities.

The periodic model with zero temperature was studied in
\cite{catto2001thermodynamic}, where it is derived as a
thermodynamic limit. The existence and uniqueness of solutions
$W(W_{\rm nucl})$ to
\eqref{eq:periodic_hartree} at finite temperature was proved in
\cite{nier1993variational}, using a variational principle for the
potential $W$. We study the convergence of fixed-point iterations to
solve these equations, both for its independent interest and to
establish the methods and estimates needed later for the study of defects.
First, for a given $W$, the charge neutrality condition can be
uniquely solved for $\varepsilon_{F}$ (see Lemma
\ref{lem:fermi_level}), yielding a map $\varepsilon_{F}(W)$ and
allowing us to reformulate the self-consistent equation as simply
$W = W_{\nucl} + \vcper F(W)$, with $F(W) = F_{\varepsilon_{F}(W)}(W)$. A
very natural iterative method to solve this equation is
\begin{align*}
  W_{n+1} = W_{\nucl} + \vcper F(W_{n}),
\end{align*}
the simple self-consistent iteration. Unfortunately, as is
well-known, this algorithm does not necessarily converge, not even
locally \cite{cances2000convergence, levitt2012convergence}. This
suggests the simple damping (or mixing) strategy
\begin{align}
  \label{eq:damped_iter_per}
  W_{n+1} = W_{n} + \alpha(W_{\nucl} + \vcper F(W_{n}) - W_{n})
\end{align}
for small $\alpha$. It is not \textit{a priori} clear why this
iteration, based on an arbitrary splitting of the self-consistent
equation, should converge, even for small $\alpha > 0$.
We prove that this is the case (recall that
$L^{2}_{\per}$ is the space of $\cR$-periodic functions that are
square-integrable over the unit cell $\Gamma$)
\begin{theorem}
  \label{thm:main_periodic}
  Assume that there is $W_{\nucl}^{*} \in L^{2}_{\per}$ and
  $W^{*} \in L^{2}_{\per}$ such that
  $W^{*} = W^{*}_{\nucl} + \vcper F(W^{*})$. Then there are
  $\alpha_{0} > 0$, neighborhoods $\mathcal W_{\nucl}$ of
  $W_{\nucl}^{*}$ and $\mathcal W$ of $W^{*}$ in $L^{2}_{\per}$ such
  that, for all $W_{\nucl} \in \mathcal W_{\nucl}$, there is a unique
  solution $W(W_{\nucl}) \in \mathcal W$ of
  $W = W_{\nucl} + \vcper F(W)$. Furthermore, for all
  $0 < \alpha \leq \alpha_{0}$, the iteration
  \eqref{eq:damped_iter_per}
  with $W_{0} \in \mathcal W$ converges to $W(W_{\nucl})$ in $L^{2}_{\per}$.
\end{theorem}
Note that the Jacobian of the fixed-point mapping
\eqref{eq:damped_iter_per} is
\begin{align*}
  J_{\alpha}(W) = 1 - \alpha  + \alpha \vcper F_{\varepsilon_{F}}'(W).
\end{align*}
We show in Lemma \ref{lem:F_per} that the Jacobian
$F_{\varepsilon_{F}}'(W)$ is bounded, self-adjoint and non-positive from
$L^{2}_{\per}$ to itself. Since $\vcper F_{\varepsilon_{F}}'(W)$ is
the product of a non-negative and a non-positive self-adjoint
operator, it has non-positive spectrum, and therefore $J_{\alpha}$
will have spectrum between $-1$ and $1$ for $\alpha$ small enough,
proving Theorem \ref{thm:main_periodic}. To analyze
$F_{\varepsilon_{F}}'(W)$, we use a contour integral formulation which
allows us to prove sum-over-states expressions for the derivatives. A
similar method was used in \cite{nier1993variational}. Although we focus on
this very simple algorithm, the behavior of more complex algorithms
such as Anderson acceleration (also known as DIIS or Pulay mixing)
depends crucially on the properties of the underlying fixed-point
iteration \cite{walker2011anderson}, and our analysis is a necessary
first step towards the understanding of these methods.

We next study defects. The model for defects for
insulators at zero temperature was introduced in \cite{cances2008new},
again through a thermodynamic limit argument. At finite temperature,
the model is as follows. We fix a solution $W_{\per}$ of the periodic
model above and its Fermi level $\varepsilon_{F}$. For a given defect
potential $V_{\defect}$, we solve the equation
\begin{align}
  \label{eq:FP_eq_def}
  V &= V_{\defect} + \vc G(V)
\end{align}
for $V$, with $G(V)$ the renormalized
potential-to-density mapping
\begin{align}
  G(V) &= \left(f_{\varepsilon_{F}}(-\Delta + W_{\per} + V)-f_{\varepsilon_{F}}(-\Delta + W_{\per})\right)(x,x).
\end{align}
Here $V$ is the total change in potential created by the addition of
the defect $V_{\defect}$. We note that, to our knowledge, neither this defect model
nor even the periodic model has been derived from a thermodynamic
limit for the rHF model at finite temperature (see
\cite{chen2018thermodynamic} for related work in a simpler model).

It is natural to try to solve
this equation by a procedure similar to \eqref{eq:damped_iter_per}:
\begin{align}
  \label{eq:damped_iter_def}
  V_{n+1} = V_{n} + \alpha(V_{\defect} + \vc G(V_{n}) - V_{n}).
\end{align}
However, in contrast to the periodic case, the operator $\vc$ is not
bounded. This is easily seen by noting that $\vc$ acts in Fourier
space as a multiplication operator by $1/|q|^{2}$. The
iteration~\eqref{eq:damped_iter_def} is therefore not well-defined.
The practical consequence of this is that, when the equations are
truncated to a finite box of linear size $L$ with appropriate boundary
conditions, $\vc$ has eigenvalues on the order of $L^{2}$. This forces
$\alpha$ to be on the order of $L^{-2}$, which slows down the
convergence\footnote{This reasoning also holds true for more complex
  methods. The Jacobian of the system has a condition number
  proportional to $L^{2}$, and therefore we expect simple methods to
  require a number of iterations proportional to $L^{2}$, and
  Krylov-type methods such as Anderson acceleration to require a
  number of iterations proportional to $L$
  \cite{walker2011anderson,saad2003iterative}.}. Because the large
eigenvalues are caused by low wavelengths, this appears in
calculations as charge moving back and forth at the extremities of the
system, a phenomenon known as charge sloshing
\cite{kerker1981efficient}. This effect does not appear when the
density is constrained to be periodic, as evidenced by Theorem
\ref{thm:main_periodic}.

This can be fixed by using a more elaborate numerical method. The
Newton method applied to \eqref{eq:FP_eq_def} is
\begin{align*}
  V_{n+1} = V_{n} + J(V_{n})^{-1} (V_{\defect} + \vc G(V_{n}) - V_{n})
\end{align*}
where
\begin{align*}
  J(V) &= 1- \vc G'(V)
\end{align*}
There is an intimate link between the Jacobian $J(V)$, describing the behavior of
iterative algorithms, and the linear response properties of the
system. The operator $\chi_{0} = G'(0)$ is called the independent-particle
susceptibility operator. It describes the linear response of the
density of a non-interacting system of electrons to a small defect
potential. It can be computed through the Adler-Wiser sum-over-states
formula \cite{adler1962quantum,wiser1963dielectric}, which we prove
in Lemma \ref{lem:chi_0_def}. The operator
$\varepsilon^{-1} = J(0)^{-1} = (1 - \vc \chi_{0})^{-1}$ is called the
dielectric operator. As we will see, it describes the linear
response of the total potential $V$ to a defect $V_{\defect}$.

Since $J(V)$ or even $J(0)$ is difficult to compute, an approximation
has to be found, yielding a preconditioned scheme. A simple
approximation can be found using the Thomas-Fermi theory of the free
electron gas \cite{lieb1977thomas}. This model takes the same form
\eqref{eq:FP_eq_def} of a fixed-point equation, but with a much simpler
potential-to-density mapping
\begin{align*}
  G_{\rm TF}(V) = (\varepsilon_{F} - V)^{\frac 3 2}_{+}
\end{align*}
where $x_{+} = \max(x,0)$. In this case we simply have
$\chi_{0,\rm TF} = G_{\rm TF}'(0) = - \frac 3 2
\sqrt{\varepsilon_{F}}$. Therefore, the operator
$\varepsilon_{\rm TF} = J_{\rm TF}(0) = 1 - \vc \chi_{0,\rm TF}$ takes
the simple form of a multiplication operator in Fourier space
\begin{align*}
  \varepsilon_{\rm TF}(q) = 1 - \frac 1 {|q|^{2}} \chi_{0,\rm TF} = \frac{|q|^{2}-\chi_{0,\rm TF}}{|q|^{2}}.
\end{align*}
The $1/|q|^{2}$ divergence for low wavelengths created by Coulomb
interaction is the cause of charge sloshing. One can then simply take
the inverse of this Thomas-Fermi Jacobian as a preconditioner. In
practice, the unknown constant $\chi_{0,\rm TF}$ is estimated
according to the system under consideration (in this paper we take it
equal to $-1$ for simplicity). This choice,
\begin{align}
  \kerker(q) = \frac{|q|^{2}}{1 + |q|^{2}},
\end{align}
or in operator form $\kerker = \frac{-\Delta}{1-\Delta}$, is known as
Kerker preconditioning \cite{kerker1981efficient}. The preconditioned
fixed-point iteration is then
\begin{align}
  \label{eq:kerker_def}
  V_{n+1} = V_{n} + \alpha \kerker (V_{\defect} + \vc G(V_{n}) - V_{n})
\end{align}
which is found in practice to substantially improve the convergence of
self-consistent algorithms.

We now turn to the related matter of screening. Expanding
\eqref{eq:FP_eq_def} to first order in $V_{\defect}$, we obtain formally
\begin{align*}
  V = (1-\vc \chi_{0})^{-1} V_{\defect} + O(\|V_{\defect}\|^{2}).
\end{align*}
As mentioned previously, the operator
\begin{align}
  \label{eq:def_dielectric}
  \varepsilon^{-1} = (1-\vc \chi_{0})^{-1}
\end{align}
is the dielectric operator. In the case of the homogeneous
Thomas-Fermi model, $\chi_{0}$ is a negative constant, and
$\varepsilon_{\rm TF}^{-1}$ is a Fourier multiplication operator given by
\begin{align*}
  \varepsilon_{\rm TF}^{-1}(q) = \frac{|q|^{2}}{|q|^{2}-\chi_{0,\rm TF}}
\end{align*}
When $V_{\defect}(x) = \frac Q{|x|}$, up to normalization we have
$\widehat{V_{\defect}}(q) = \frac Q{|q|^{2}}$, and so
\begin{align*}
  \widehat{\varepsilon_{\rm TF}^{-1} V_{\defect}}(q) = \frac Q {|q|^{2} - \chi_{0,\rm TF}},
\end{align*}
the Fourier transform of a short-range Yukawa potential
\begin{align*}
  (\varepsilon_{\rm TF}^{-1} V_{\defect})(x) = Q \,\frac{ e^{-\sqrt{-\chi_{0,\rm TF}} |x|}}{|x|}.
\end{align*}
The Thomas-Fermi theory of screening beyond linear response was discussed in
\cite{lieb1977thomas}, and extended to the Thomas--Fermi--von
Weisz{\"a}cker model in \cite{cances2011local,nazar2017locality}.

The purpose of this paper is to extend the justification of Kerker
preconditioning as well as the Thomas-Fermi theory of screening to the
more realistic rHF model of defects.

Our main result is
\begin{theorem}
  \label{thm:existence_defect}
  Fix $W_{\per} \in L^{2}_{\per}$ and $\varepsilon_{F} \in \R$. There are $\alpha_{0} > 0$ and neighborhoods
  $\mathcal V_{\defect}$ and $\mathcal V$ of $0$ in $\vc H^{-2}$
  and $L^{2}$ respectively such that, for all $V_{\defect} \in \mathcal V_{\defect}$,
  there is a unique solution
  $V(V_{\defect})$ of
  \begin{align*}
    V = V_{\defect} + \vc G(V)
  \end{align*}
  in $\mathcal V$. Furthermore, for $0 < \alpha \leq \alpha_{0}$, the
  iteration
  \begin{align*}
    V_{n+1} = V_{n} + \alpha \kerker (V_{\defect} + \vc G(V_{n}) - V_{n})
  \end{align*}
  with $V_{0} \in \mathcal V$ converges to $V(V_{\defect})$ in
  $L^{2}$.

  We have the expansion
  \begin{align*}
    V(V_{\defect}) = \varepsilon^{-1} V_{\defect} + O(\|V_{\defect}\|_{\vc H^{-2}}^{2})
  \end{align*}
  in $L^{2}$, where
  \begin{align*}
    \varepsilon^{-1} = (1 - \vc \chi_{0})^{-1}
  \end{align*}
  is continuous from $\vc H^{-2}$ to $L^{2}$, and $\chi_{0} =
  G'(0)$ is continuous from $L^{2}$ to itself.
\end{theorem}
Here the space
\begin{align*}
  \vc H^{-2} = \{\vc f, f \in H^{-2}\} = \left\{f, \int_{\R^{3}}|\widehat f(q)|^{2} \frac {|q|^{4}}{(1 + |q|^{2} + |q|^{4})} dq < \infty\right\}
\end{align*}
is large enough to contain point defect
potentials of the form $V_{\defect}(x) = \frac {Q}{|x|}$. In this
case, our theorem states that when $Q$ is small enough, the screened
potential $V(V_{\defect})$ is in $L^{2}$, and therefore decays faster
than $V_{\defect}$. When the defect potential is the Coulomb potential
generated by a localized charge density $\rho$, we expect from the
analysis of the Thomas-Fermi model that $V(V_{\defect})$ will have the
same decay properties as $\rho$ (because
$q \mapsto \frac{\varepsilon_{\rm TF}^{-1}(q)}{|q|^{2}}$ is smooth).
To quantify this, we define the weighted Lebesgue and Sobolev spaces
(see Section \ref{sec:notations} for more details): for every
$n \in \R, N \geq 0$,
\begin{align*}
  L^{2}_{N} = \left\{f \in L^{2}, \int_{\R^{3}} (1+|x|^{2})^{N} |f(x)|^{2} dx < \infty\right\}
\end{align*}
and
\begin{align*}
  H^{n}_{N} = \{f, (1+|x|^{2})^{\frac N 2} f \in H^{n}\}.
\end{align*}
We then have
\begin{theorem}
  \label{thm:decay}
  Fix $W_{\per} \in L^{2}_{\per}$ and $\varepsilon_{F} \in \R$. There
  is a neighborhood $\widetilde{\mathcal V_{\defect}} \subset \mathcal
  V_{\defect}$ of zero in $\vc H^{-2}_{1}$
  such that, if $V_{\defect} \in \widetilde{\mathcal V_{\defect}}$, and if $V_{\defect} \in
  \vc H^{-2}_{N}$, then $V(V_{\defect}) \in L^{2}_{N}$.
\end{theorem}
Therefore, if $V_{\defect}(x) = \frac Q {|x|}$ for $Q$ small enough,
then $V(V_{\defect})$ decays faster than any polynomial.

To prove Theorem \ref{thm:existence_defect}, we need to generalize the results of the
Thomas-Fermi model to our setting. The first obstacle is the more
complicated nature of the potential-to-density mapping $G$. This is
handled by using a contour-integral formulation, which allows for the
computation of response functions (derivatives of $G$). The second is
the absence of translation invariance, and therefore of the simple
decomposition of operators in Fourier space. However, the periodicity of
the underlying crystal allows the use of the Bloch transform, which
replaces the Fourier transform used in the homogeneous case. We also
need to establish the invertibility of the operator
$\varepsilon \kerker$, which is done by studying the low-wavelength
behavior of the independent-particle susceptibility operator
$\chi_{0}$, and relating it to $F_{\varepsilon_{F}}'(W_{\per})$.
Finally, the improved decay estimates in Theorem \ref{thm:decay} are
obtained by considering the off-diagonal decay of the resolvent of the
periodic Hamiltonian, a property related to the well-known locality of
the density matrix \cite{prodan2005norm,benzi2013decay, cances2018numerical}.
\begin{remark}[Exponential decay]
  It follows from our estimates that the operator
  $\varepsilon^{-1} \vc$ representing the linear response of the
  screened potential to a defect charge density has an exponentially
  decaying kernel. Indeed, from the proof of Lemma~\ref{lem:chi_0_def}
  one can see that its fibers are analytic in a strip in the complex
  plane, and therefore $\varepsilon^{-1} \vc$ maps exponentially
  decaying charge densities to exponentially decaying potentials. The
  exponential decay rate depends in particular on the temperature.
  Proving this for the non-linear mapping $V(V_{\defect})$ requires
  the use of more involved functional spaces quantifying exponential
  decay, and we do not do it in this paper.
\end{remark}
\begin{remark}[Zero temperature limit]
  The results above are to be compared with those of
  \cite{cances2010dielectric} (see also \cite{cances2012mathematical}
  for the dynamical case). There, the authors study the linear
  response in the case of insulators at zero temperature. They obtain
  partial screening, whereby the total potential behaves at long range
  as a Coulombic potential whose effective charge is reduced by a
  constant factor (the dielectric constant of the material). The
  difference can be schematized as follows: in the case of insulators
  at zero temperature, the independent-particle susceptibility
  operator $\chi_{0}$ behaves for low wavelengths as $|q|^{2}$,
  reflecting the lack of bulk movement of electrons.
  Accordingly, the dielectric operator
  $\varepsilon^{-1} = (1-\vc \chi_{0})^{-1}$ behaves as a constant. In the
  finite-temperature case, $\chi_{0}$ behaves as a constant for low
  wavelengths, and therefore $\varepsilon^{-1}$ behaves as $|q|^{2}$.

  The discussion above in terms of wavelengths is complicated by the
  fact that these operators do not commute with all translations but
  only with those of the crystal lattice, and so are not diagonalized
  by the Fourier transform but by the Bloch transform. Because of
  the appearance of the inverse, the behavior of $\varepsilon^{-1}$
  for low wavelengths is not determined only by that of $\chi_{0}$ for low
  wavelengths. This discrepancy is sometimes called ``local field
  effects'' in the physical literature. However, the conclusions above
  are qualitatively correct, although the proper treatment of these
  effects is more involved, as we will see.

  This work is only concerned with the finite-temperature case.
  Physically, this has the effect of making every material metallic,
  in the sense that there are free electrons available to move towards
  the defect. Mathematically, this allows response functions to be
  derived straightforwardly from contour integrals. The case of the
  zero-temperature limit of metals remains open (although see
  \cite{frank2013positive} in the linear case). A particular challenge
  is that of the appearance of Friedel oscillations, which in the case
  of the free Fermi gas ($W_{\per} = 0$) are linked with
  non-smoothness of the independent-particle susceptibility
  $\chi_{0}(q)$. In the periodic case, the shape of Friedel
  oscillations depends on the properties of the Fermi surface.
\end{remark}

\begin{remark}[Energy methods]
  \label{rem:energy}
  In this work, we are concerned with the convergence of fixed-point
  iterations, and screening in the small defect regime. Therefore, we
  use a fixed-point approach to the existence of solutions of the
  defect equations, and do not exploit the existence of an energy.
  This limits our range of applicability to small defects, and cannot
  ensure the uniqueness of solutions. It would be interesting to prove
  the existence of solutions outside of the perturbative regime
  through energy methods.

  The use of an energy sheds some light on the convergence of the
  damped fixed-point iteration, which decreases the (free) energy of
  the system for small enough damping parameter. Similarly, the
  non-positivity of the derivative of the potential-to-density
  mapping, which we obtained by direct computation, can also be seen
  through energy methods. For concreteness, we sketch this argument
  now in a periodic system at fixed Fermi level. Consider a periodic
  system of non-interacting electrons in a periodic potential $W$.
  Define the free energy (per unit cell) of a density matrix $\gamma$
  \begin{align*}
    E(\gamma, W) = \underline{\Tr}((-\Delta + W - \varepsilon_{F}) \gamma) + k_{B}T\, \underline{\Tr}(\gamma \log \gamma + (1-\gamma) \log(1-\gamma))
  \end{align*}
  where $\underline{\Tr}$ is the trace per unit cell (see Section
  \ref{sec:notations}). Then $E$ is convex on a suitable subset of the
  convex set of periodic self-adjoint operators satisfying
  $0 \leq \gamma \leq 1$ and admits a unique minimizer
  $\gamma^{*}(W) = f_{\varepsilon_{F}}(-\Delta + W)$. The functional
  \begin{align*}
    I(W) = \inf_{\gamma} E(\gamma,W) = E(\gamma^{*}(W),W)
  \end{align*}
  is concave, being the infimum of affine functionals. Its gradient is
  computed using an Hellmann-Feynman-type argument as
  \begin{align*}
    I'(W) = \frac{\partial E}{\partial W}(\gamma^{*}(W),W) = \gamma^{*}(W)(x,x) = F_{\varepsilon_{F}}(W)
  \end{align*}
  and it follows that $F_{\varepsilon_{F}}'$, being the Hessian of a
  concave functional, is self-adjoint and non-positive.
\end{remark}

\begin{remark}[Kohn-Sham density functional theory]
  We consider here the rHF model, which neglects any
  exchange-correlation effects. In the case of the Kohn-Sham model
  under the local density approximation (LDA), the equation becomes
  $V = V_{\rm def} + \vc G(V) + V_{\rm xc}(G(V))$ where
  $V_{\rm xc}(\rho)$ is the exchange-correlation potential (the
  gradient of the exchange-correlation energy). The dielectric
  operator is then
  \begin{align*}
    \varepsilon^{-1} = (1 - (\vc+K_{\rm xc}) \chi_{0})^{-1},
  \end{align*}
  where $K_{\rm xc} = V_{\rm xc}'$. Crucially, $K_{\rm xc}$ is not in
  general a positive operator, since the exchange-correlation energy
  is not convex. It is then not a priori clear that the operator
  $1 - (\vc+K_{\rm xc}) \chi_{0}$ is invertible, even for a finite
  system. This property however holds at a non-degenerate local
  minimum of the energy \cite{cances2020convergence}. The
  investigation of screening in the Kohn-Sham model under this
  condition would be interesting future work.
\end{remark}

The structure of the paper is as follows. We first introduce our
notations in Section \ref{sec:notations} and recall properties of the Bloch
transform and of periodic operators. In Section \ref{sec:estimates} we
state general theorems and prove some estimates on resolvents and
densities of operators. Then we study the
periodic rHF model in Section \ref{sec:periodic}, establishing
properties of the response operators and proving Theorem
\ref{thm:main_periodic}. We finally study the defect
model in Section \ref{sec:defect}, culminating in the proof of Theorems
\ref{thm:existence_defect} and \ref{thm:decay}.

\section{Notations}
\label{sec:notations}
Let $\cR$ be a periodic lattice in $\R^{3}$, $\cR^{*} = \{K \in
\R^{3}|\, \forall R \in \cR, e^{iK \cdot R} = 1\}$ be its dual
lattice, $\Gamma$ be a unit cell of $\cR$, and $\Gamma^{*} = \BZ$ be a
unit cell of $\cR^{*}$. By abuse of language we call $\BZ$ the
Brillouin zone. Both $\Gamma$ and $\BZ$ are considered to have the
topology of a torus: this means that, for instance, a continuous
function on $\Gamma$ extends to a continuous and $\cR$-periodic
function on $\R^{3}$.

We let $k_{B}T > 0$ be a fixed
temperature, and set
\begin{align*}
  f_{\varepsilon_{F}}(\varepsilon) = \frac 1 {1+e^{\frac {\varepsilon-\varepsilon_{F}} {k_{B} T}}}
\end{align*}
the Fermi-Dirac occupation function. We recall that $f$ is decreasing
on $\R$ and analytic on $\R + i (-\pi k_{B} T, \pi k_{B} T)$.

$L^{2}$ is the usual Lebesgue space on $\R^{3}$,
and $L^{2}_{\per} \sim L^{2}(\Gamma)$ is the space of $\cR$-periodic
functions. For $s \in \R$, $H^{s}$ is the Sobolev space on $\R^{3}$
and $H^{s}_{\per} \sim H^{s}(\Gamma)$ the Sobolev space on the torus
$\Gamma$, defined via Fourier transform and Fourier series
respectively. All these spaces are Hilbert spaces with their usual
inner product.

We normalize the Fourier series, transforms and Bloch transforms to
consistently have un-normalized decompositions: for a function
$u \in L^{2}_{\per}$, we have
\begin{align*}
  u(x) = \sum_{K \in \cR^{*}} e^{i Kx} c_{K}(u), \quad c_{K}(u) = \fint_{\Gamma} e^{-iKx} u(x) dx
\end{align*}
where $\fint_{\Omega} = \frac 1 {\Omega}\int_{\Omega}$ is the normalized integral. For a
function $w \in L^{2}$ we have
\begin{align*}
  w(x) = \int_{\R^{3}} e^{iqx} \widehat w(q) dq , \quad \widehat w(q) = \frac 1 {(2\pi)^{3}} \int_{\R^{3}} e^{-iqx} w(x) dx
\end{align*}
The Bloch transform for $w \in L^{2}$ is
\begin{align*}
  w(x) = \int_{\BZ} e^{ikx} u_{k}(x) dk, \quad u_{k}(x) = \sum_{K \in \cR^{*}} e^{iKx} \widehat w(k+K).
\end{align*}
The map $k \mapsto u_{k}$ belongs to the space
$L^{2}(\BZ, L^{2}_{\per})$, by which we mean the space of functions
$u_{k} : \R^{3} \mapsto L^{2}_{\per}$ that are locally $L^{2}$ and
satisfy the pseudo-periodicity condition
$u_{k+K}(x) = e^{-iKx} u_{k}(x)$ for all $K \in \cR^{*}$. This space is equipped
with the norm
\begin{align*}
  \|u\|_{L^{2}(\BZ, L^{2}_{\per})}^{2} = \int_{\BZ} \|u_{k}\|_{L^{2}_{\per}}^{2}.
\end{align*}
The Bloch transform is, up to normalization, unitary from $L^{2}$ to
$L^{2}(\BZ, L^{2}_{\per})$.

Recall that $-i\partial_{x_{i}} w$ has Bloch transform
$k \mapsto (-i\partial_{x_{i}} + k_{i}) u_{k}$, and that $x_{i} w$ has
Bloch transform $k \mapsto i \partial_{k_{i}} u_{k}$. Let
$\langle x \rangle = \sqrt{1+|x|^{2}}$. For every $n, N \in \R$, let
the weighted Sobolev spaces
\begin{align*}
  H^{n}_{N} = \{f \in \mathcal S', \langle x \rangle^{N} f \in H^{n}\}
\end{align*}
and
\begin{align*}
  L^{2}_{N} = H^{0}_{N} = \{f \in L^{2}, \|\langle x\rangle^{N} f\|_{L^{2}} < \infty\}
\end{align*}
equipped with their natural inner products. Here $\mathcal S'$ is the space of tempered distributions. The Fourier transform is bounded and invertible
from $H^{n}_{N}$ to $H^{N}_{n}$. The Bloch transform is similarly
bounded and invertible from $H^{n}_{N}$ to $H^{N}(\BZ, H^{n}_{\per})$,
where $H^{N}(\BZ, H^{n}_{\per})$ is defined as above (see for instance
\cite{lechleiter2017floquet}). 

If $A$ is a bounded operator on a Banach space, we call $\|A\|$ its
norm, $\sigma(A)$ its spectrum and $r(A) =
\lim_{n\to\infty}\|A^{n}\|^{1/n} = \sup\{|z|, z \in \sigma(A)\}$ its spectral
radius.

We denote by $\Schatten_{p}$ the space of Schatten-class operators on
$L^{2}$. The spaces $\Schatten_{p}$ equipped with their norm
$\|A\|_{\Schatten_{p}} = (\Tr |A|^{p})^{1/p}$ are Banach spaces
(Hilbert space for $p=2$). In particular, the cases $p=1,2,\infty$
correspond to trace-class, Hilbert-Schmidt and bounded operators
respectively.

We say that a bounded operator $A$ on $L^{2}(\R^{3})$ is a periodic
operator if it commutes with the translations of the lattice $\cR$. As
is well-known \cite{reed1978analysis}, such operators are decomposed by the Bloch transform,
in the sense that there exists a family $\{A_{k}\}_{k \in \BZ}$ of
bounded operators on $L^{2}_{\per}$ such that, if
$w = \int_{\BZ} e^{ikx} u_{k}(x) dk \in L^{2}$, then
\begin{align*}
  (A w)(x) = \int_{\BZ} e^{ikx} (A_{k} u_{k})(x) dk.
\end{align*}
We call the operators $A_{k}$ the fibers of $A$. The smoothness of the
fibers of operators reflect the off-diagonal properties of their
kernel: if an operator $A$ has fibers $A_{k}$ that are smooth from
$\R^{3}$ to bounded operators from $H^{n}_{\per}$ to $H^{m}_{\per}$
and if $w \in H^{n}_{N}$ for some $N \geq 0$, then $A w \in H^{m}_{N}$.

If $A_{k}$ are trace-class on $L^{2}_{\per}$ almost everywhere and
$\int_{\BZ} \Tr |A_{k}| < \infty$, we define the trace per unit cell
\begin{align*}
  \Trper A = \fint_{\BZ} \Tr A_{k}.
\end{align*}
One can then define the Schatten classes of periodic operators
\begin{align*}
  \Schatten_{p,\per} = \{A \text{ periodic}, \Trper |A|^{p} < \infty\}
\end{align*}
with associated norms. Note that this is distinct from (and larger
than) the class of Schatten operators on $L^{2}_{\per}$.

If $A \in \Schatten_{1}$, then $A$ has the singular value decomposition
$A = \sum_{n \in \N} \lambda_{n} |\phi_{n}\rangle\langle \psi_{n}|$
with $\phi_{n}$ and $\psi_{n}$ two orthonormal sets and
$\sum_{n \in \N} |\lambda_{i}| < \infty$, and we define its density
$A(x,x) \in L^{1}$ by
\begin{align*}
  A(x,x) = \sum_{n \in \N} \lambda_{n} \phi_{n}(x)
\overline{\psi_{n}}(x).
\end{align*}
Similarly, if $A$ is locally trace class then
$A(x,x) \in L^{1}_{\rm loc}$, if $A$ is a trace-class operator on
$L^{2}_{\per}$ then $A(x,x) \in L^{1}_{\per}$ and if $A$ is in
$\Schatten_{1,\per}$ then $A(x,x) \in L^{1}_{\per}$, with
\begin{align*}
  A(x,x) = \fint_{\BZ} A_{k}(x,x) dk.
\end{align*}

\section{General results and estimates}
\label{sec:estimates}
\subsection{General results}
We recall the following classical properties
\begin{lemma}
  \label{lem:sigma_AB}
    Let $X$ be a Banach space and $A$, $B$ be bounded operators on
    $X$. Then $\sigma(AB) \setminus \{0\} = \sigma(BA)\setminus \{0\}$.
  \end{lemma}
  \begin{proof}
    Let $\lambda \notin \sigma(AB)$ and $\lambda \neq 0$. Then
    $(\lambda - BA)$ is invertible with inverse
    \begin{align*}
      (\lambda-BA)^{-1} =
      \lambda^{-1} (1 + B (\lambda-AB)^{-1} A)
    \end{align*}
    and $\lambda \notin \sigma(BA)$. The proof follows by
    interchanging $A$ and $B$.
  \end{proof}
  
\begin{lemma}
  \label{lem:norm_change}
  Let $X$ be a Banach space and $A$ a bounded operator on $X$. Then
  for every $\varepsilon > 0$, there is a norm
  $\|\cdot\|_{\varepsilon}$ equivalent to $\|\cdot\|_{X}$ such that
  $\|A\|_{\varepsilon} \leq r(A) + \varepsilon$.
\end{lemma}
\begin{proof}
  See \cite{holmes1968formula} for instance.
\end{proof}

We will make use of the following variant of the Banach fixed point theorem:
  \begin{theorem}
    \label{thm:abstract_FP}
  Let $X,Y$ be two Banach spaces, $U$ and $V$ be two neighborhoods of
  $x^{*} \in X$ and $y^{*} \in Y$, and $M : U \times V \mapsto X$ be
  a continuously differentiable mapping such that $M(x^{*},y^{*}) = x^{*}$,
  and
  \begin{align*}
    r\left(\frac{\partial M}{\partial x}(x^{*},y^{*})\right) < 1.
  \end{align*}
  Then there are neighborhoods $\widetilde U \subset U$ and
  $\widetilde V \subset V$ of $x^{*}$ and $y^{*}$ such that, for all $y \in \widetilde V$,
  the iteration
  \begin{align}
    \label{eq:abstract_FP}
    x_{n+1} = M(x_{n},y)
  \end{align}
  with $x_{0} \in \widetilde U$ converges to a solution $x(y)$ of $M(x(y),y) = x(y)$ in
  $\widetilde U$. This solution is unique in $\widetilde U$.
  Furthermore, $x(y)$ is differentiable, and
  \begin{align*}
    x'(y) = \left(1-\frac{\partial M}{\partial x}(x(y),y)\right)^{-1} \frac{\partial M}{\partial y}(x(y),y).
  \end{align*}
\end{theorem}
\begin{proof}
  Applying Lemma \ref{lem:norm_change} to
  $A=\frac{\partial M}{\partial x}(x^{*},y^{*})$ and using the
  continuous differentiability of $M$, we see that, for all
  $\varepsilon > 0$, there is an equivalent norm
  $\|\cdot\|_{\varepsilon}$ on $X$ such that
  \begin{align*}
    \left\|\frac{\partial M}{\partial x}(x,y)\right\|_{\varepsilon} \leq r\left(\frac{\partial M}{\partial x}(x^{*},y^{*})\right) + \varepsilon + O(\|x-x^{*}\|+\|y-y^{*}\|)
  \end{align*}
  It follows that, for $\varepsilon$ small enough, there is a neighborhood
  $\widetilde U \times \widetilde V$ of $(x^{*},y^{*})$ such that, for
  every $y \in \widetilde V$, $M(\cdot,y)$ maps $\widetilde U$ to
  itself and is a contraction for the $\|\cdot\|_{\varepsilon}$ norm. The
  convergence of \eqref{eq:abstract_FP} (in the
  $\|\cdot\|_{\varepsilon}$ and therefore in the $\|\cdot\|$ norm), as
  well as the uniqueness of $x(y)$ follows from the Banach fixed-point
  theorem. The differentiability follows as in the proof of the
  implicit function theorem.
\end{proof}
\begin{remark}
  The implicit function theorem also shows the existence of $x(y)$
  under weaker assumptions (that
  $1-\frac{\partial M}{\partial x}$ is invertible). The main
  difference is that the implicit function theorem uses the
  Newton-like iteration
  $x_{n+1} = x_{n} + (1-\frac{\partial M}{\partial
    x}(x^{*},y^{*}))^{-1} (M(x_{n},y)-x_{n})$ instead of the simpler
  iteration \eqref{eq:abstract_FP}. We use here this version because
  we are interested in the convergence of the fixed-point iteration. 
  
  Recall that in general
  $r(\frac{\partial M}{\partial x}) < \|\frac{\partial M}{\partial
    x}\|$ for general non-normal operators, and therefore $M$ is not
  necessarily a contraction.
\end{remark}

\subsection{Resolvent estimates}
In the following, we want to prove that products of resolvents and
potentials have certain trace properties, in order to
define potentials-to-density mappings via contour integrals. The
following equality, a building block of the general Kato-Seiler-Simon
inequality \cite{simon2010trace}, will be very useful:
\begin{lemma}[Kato-Seiler-Simon equality]
  For every $f \in L^{2}$, $g \in L^{2}$, $f(-i\nabla)g(x) \in
  \Schatten_{2}$ and
  \begin{align*}
    \|f(-i \nabla) g(x)\|_{\Schatten_{2}} = (2\pi)^{-3/2}\|f\|_{L^{2}} \|g\|_{L^{2}}.
  \end{align*}

  Similarly, for $f \in L^{2}, g \in L^{2}_{\per}$,
  \begin{align*}
    \|f(-i \nabla) g(x)\|_{\Schatten_{2,\per}} = (2\pi)^{-3/2} \|f\|_{L^{2}} \|g\|_{L^{2}_{\per}}.
  \end{align*}
\end{lemma}
\begin{proof}
  The proof of the first assertion is standard, see e.g.
  \cite{simon2010trace}: note that $f(-i\nabla) g(x)$ has integral
  kernel
  \begin{align*}
    (f(-i\nabla) g(x))(x,y) = \check{f}(x-y)g(y)
  \end{align*}
  and therefore
  \begin{align*}
    \|f(-i\nabla) g(x)\|_{\Schatten_{2}}^{2} = \int_{\R^{6}}|\check{f}(x-y)g(y)|^{2} dx dy = (2\pi)^{-3}\|f\|_{L^{2}} \|g\|_{L^{2}}
  \end{align*}
  For the second, we first note that, if $f \in \ell^{2}(\cR^{*})$ and
  $g \in L^{2}_{\per}$, then $f(-i\nabla) g(x)$ is an operator on
  $L^{2}_{\per}$. By writing its kernel (matrix in the
  basis of the $\frac 1 {\sqrt{|\Gamma|}}e^{iKx}$ for $K \in \cR^{*}$),
  we get
  \begin{align*}
    \|f(-i\nabla) g(x)\|_{\Schatten_{2}(L^{2}_{\per})} = \frac{1}{\sqrt{|\Gamma|}} \|f\|_{\ell^{2}(\cR^{*})} \|g\|_{L^{2}_{\per}}
  \end{align*}
  (see \cite{gontier2016supercell}). The result then follows by
  writing, for $f \in L^{2}, g \in L^{2}_{\per}$,
  \begin{align*}
    \|f(-i\nabla) g(x)\|_{\Schatten_{2,\per}}^{2} &= \fint_{\BZ} \|f(-i\nabla+k) g(x)\|_{\Schatten_{2}(L^{2}_{\per})}^{2} dk\\
    &= \frac 1 {|\Gamma|} \|g\|_{L^{2}_{\per}}^{2}\fint_{\BZ} \sum_{K \in \cR^{*}} |f(k+K)|^{2} dk= (2\pi)^{-3} \|g\|_{L^{2}_{\per}}^{2}\|f\|_{L^{2}}^{2}
  \end{align*}
\end{proof}

In particular, since $f(q) = (1+|q|^{2})^{-1}$ is in $L^{2}$, this
implies that $(1-\Delta)^{-1} V$ is bounded for
$V \in L^{2} + L^{2}_{\per}$. This can be amplified to prove that both
$L^{2}$ and $L^{2}_{\per}$ potentials are $-\Delta$-bounded with
relative bound zero, so that, for $V \in L^{2} + L^{2}_{\per}$,
$-\Delta + V$ is self-adjoint on $L^{2}(\R^{3})$ with domain
$H^{2}(\R^{3})$. In particular, the resolvent of $-\Delta + V$ is the
resolvent of the Laplacian, modulo a bounded operator:
\begin{lemma}
  \label{lemma:res-inv-lap}
  There is $C > 0$ such that, for all $V \in L^{2} + L^{2}_{\per}$, if
  $z \not \in \sigma(-\Delta + V)$, then
  \begin{align}
    B_{z} = (z-H)^{-1} (1-\Delta)
  \end{align}
is bounded, with
  \begin{align*}
    \|B_{z}\| \leq C \left( 1+\frac{1+|z|+ \|V\|_{L^{2}+L^{2}_{\per}}^{4}}{d(z,\sigma(H))} \right).
  \end{align*}
\end{lemma}
\begin{proof}
  In this proof and others in the sequel, $C$ denotes a constant whose
  value might change from line to line.
  
  The following argument is classical, see e.g. \cite[Lemma
  1]{cances2008new}. The idea of the proof is that, if $V$ is small,
  we can expand
  $(z-H)^{-1}(1-\Delta) = \sum_{n \geq 0}
  ((z+\Delta)^{-1}V)^{n}(z+\Delta)^{-1}(1-\Delta)$ and bound
  $(z+\Delta)^{-1}V$ by the Kato-Seiler-Simon equality. To extend
  this argument for arbitrary large sizes of $V$, we consider the
  shifted operator $H+ic$, where $c > 0$.

  Let $c > 0$. For $V \in L^{2}$, we have that
  \begin{align*}
    \|(- \Delta - ic)^{-1} V\| \leq \|(- \Delta - ic)^{-1} V\|_{\Schatten_{2}} \leq \|V\|_{L^{2}} \|(|\xi|^{2} - ic)^{-1}\|_{L^{2}} \leq c^{-1/4}\|V\|_{L^{2}} \|(|\xi|^{2} - i)^{-1}\|_{L^{2}}
  \end{align*}
  while similarly for $W \in L^{2}_{\per}$ we have
  \begin{align*}
    \|(- \Delta - ic)^{-1} W\| \leq \|(- \Delta - ic)^{-1} W\|_{\Schatten_{2,\per}} \leq \|W\|_{L^{2}_{\per}} \|(|K|^{2} - ic)^{-1}\|_{\ell^{2}} \leq c^{-1/4}\|W\|_{L^{2}_{\per}} \|(|K|^{2} - i)^{-1}\|_{\ell^{2}}
  \end{align*}
  
  It follows that by taking
  $c = C (1+|z|+ \|V\|_{L^{2}+L^{2}_{\per}}^{4})$ with $C$ large
  enough, we get
  \begin{align*}
    (z-(H+ic))^{-1} = ((-\Delta-ic) +z- V)^{-1} = (1 + (-\Delta-ic)^{-1}(z-V))^{-1} (-\Delta-ic)^{-1}
  \end{align*}
  and so $(z-(H+ic))^{-1}(1-\Delta)$ is bounded uniformly in $V$ and
  $z$. The result then follows from
  \begin{align*}
    B_{z} &= (z-H)^{-1}(z-(H+ic)) (z-(H+ic))^{-1} (1-\Delta)\\
    \|B_{z}\| &\leq C  \sup_{\lambda \in \R} \frac{|z-(\lambda+ic)|}{z - \lambda} \leq C\left(1 + \frac c {d(z,\sigma(H))}\right)
  \end{align*}
\end{proof}

\subsection{Density of an operator}
\label{sec:def_density}
The following lemma gives a useful condition for an operator to have a
density in $L^{2}$, or for a periodic operator to have a density in $L^{2}_{\per}$.
\begin{lemma}
  \label{lem:def_density}
  There is $C > 0$ such that, if $A$ is an operator such that $A(1-\Delta) \in
  \Schatten_{2}$, then $A(x,x) \in L^{2}$, with
  \begin{align*}
    \|A(x,x)\|_{L^{2}} \leq C \|A(1-\Delta)\|_{\Schatten_{2}}
  \end{align*}
  Similarly, if $A_{\per}$ is a periodic operator such that
  $A_{\per}(1-\Delta) \in \Schatten_{2,\per}$, then $A_{\per}(x,x) \in
  L^{2}_{\per}$, with
  \begin{align*}
    \|A_{\per}(x,x)\|_{L^{2}_{\per}} \leq C \|A_{\per}(1-\Delta)\|_{\Schatten_{2,\per}}.
  \end{align*}
\end{lemma}
\begin{proof}
  For any function $f \in L^{2}$,
  \begin{align*}
    \int_{\R^{3}} A(x,x) f(x) dx = \Tr(Af) \le \|A f\|_{\Schatten_{1}} \leq \|A(1-\Delta)\|_{\Schatten_{2}} \|(1-\Delta)^{-1}f\|_{\Schatten_{2}} \leq C \|A(1-\Delta)\|_{\Schatten_{2}} \|f\|_{L^{2}},
  \end{align*}
  where $f$ above is interpreted as a multiplication operator. The proof is similar in
  the periodic case.
\end{proof}
\section{The periodic finite-temperature rHF model}
\label{sec:periodic}
Given a nuclear potential $W_{\nucl} \in L^{2}_{\per}$, we look for a solution of
the equations
\begin{align}
  \label{eq:RHF_per}
  \begin{cases}
    W = W_{\nucl} + \vcper F_{\varepsilon_{F}}(W)\\
    \int_{\Gamma} F_{\varepsilon_{F}}(W) = N_{\rm el}.
  \end{cases}
\end{align}
Recall that the existence and uniqueness of solutions of this equation
have been proved in \cite{nier1993variational}. Our goal for this
section is Theorem \ref{thm:main_periodic}, which states the local
convergence of a fixed-point iteration.

For any $\rho \in L^{2}_{\per}$, $\vcper \rho$ was defined in \eqref{eq:vcper} as the
solution of the periodic Poisson equation with zero mean:
\begin{align*}
  (\vcper \rho)(x) = \sum_{K \in \cR^{*}, K \neq 0} \frac{c_{K}(\rho)}{|K|^{2}} e^{iKx}.
\end{align*}
It is a bounded non-negative self-adjoint operator on $L^{2}_{\per}$.
It is the pseudo-inverse of the negative Laplacian on $L^{2}_{\per}$,
in the sense that $-\Delta (\vcper \rho) = \rho$ for all
$\rho \in L^{2}_{\per}$ with $\rho \perp e$, and $\vcper e = 0$, where
the constant function $e(x) = 1$ spans the kernel of $-\Delta$.

We first investigate the mapping $F_{\varepsilon_{F}}$ and its derivative. The last
property that $F'_{\varepsilon_{F}}(W) + \beta \Delta$ is positive
for all $\beta > 0$ is recorded for future use in the case of defects.
\begin{lemma}
  \label{lem:F_per}
  For all $\varepsilon_{F} \in \R$, the map
  \begin{align*}
    F_{\varepsilon_{F}}(W) = f_{\varepsilon_{F}}\left(-\Delta + W\right)(x,x)
  \end{align*}
  is analytic from $L^{2}_{\per}$ to itself. For all
  $W \in L^{2}_{\per}$, its differential $F_{\varepsilon_{F}}'(W)$ is
  self-adjoint and non-positive. Furthermore, for every
  $\beta > 0$, $ F'_{\varepsilon_{F}}(W) + \beta \Delta$ is
  negative.
\end{lemma}

\begin{proof}
  \textbf{Step 1:
    $F_{\varepsilon_{F}} : L^{2}_{\per} \to L^{2}_{\per}$.} Let
  $W \in L^{2}_{\per}$, and $H = -\Delta + W$. Recall that $H$ is
  periodic, with fibers $H_{k} = (-i\nabla + k)^{2} + W$. We label the
  eigenvectors and eigenvalues of $H_{k}$ (a self-adjoint operator on
  $L^{2}_{\per}$ with compact resolvent) by
  \begin{align*}
    H_{k} u_{nk} = \varepsilon_{nk} u_{nk}
  \end{align*}
  where the $(\varepsilon_{nk})_{n \in \N}$ are ordered by increasing
  order. We have
  \begin{align*}
    F_{\varepsilon_{F}}(W) = \fint_{\BZ} \sum_{n \in \N} f_{\varepsilon_{F}}(\varepsilon_{nk}) |u_{nk}|^{2} dk
  \end{align*}
  By standard comparison arguments, there are $a \in \R, b > 0$ such
  that $\varepsilon_{nk} \geq a+b n^{2/3}$. By the Sobolev embedding
  $H^{1}_{\per}\hookrightarrow L^{4}_{\per}$, $|u_{nk}|^{2}$ is
  controlled in $L^{2}_{\per}$ by
  $\|u_{nk}\|_{H^{1}_{\per}}^{2} \leq C(1+ n^{2/3})$ for some $C > 0$,
  uniformly in $k \in \BZ$, and it follows from the exponential decay
  of $f_{\varepsilon_{F}}$ that
  $F_{\varepsilon_{F}}(W) \in L^{2}_{\per}$.

  \begin{figure}[h!]
    \centering
    \includegraphics[width=.5\textwidth]{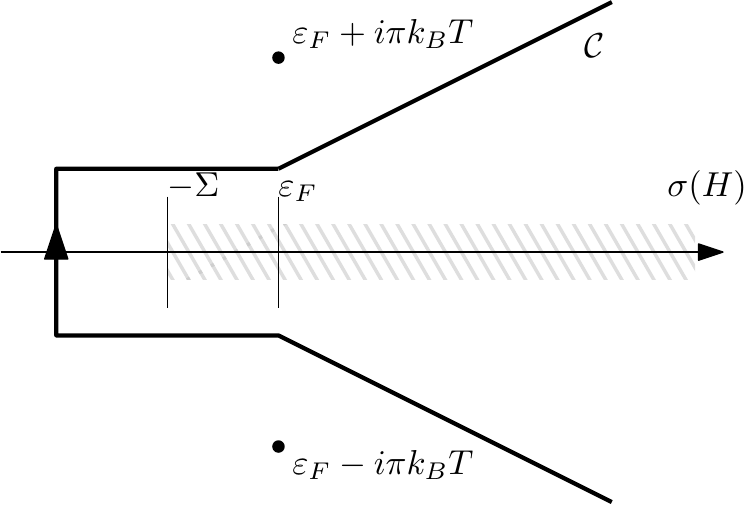}
    \caption{Contour $\cont$ used to differentiate the potential-to-density
      mapping. Note that this differs from standard rectangular
      contours because we need to ensure that $(z-H)^{-1}$
      remains Hilbert-Schmidt uniformly in $z \in \cont$.}
    \label{fig:cont}
  \end{figure}
  
  \textbf{Step 2: $F_{\varepsilon_{F}}$ is analytic.} Since potentials
  in $L^{2}_{\per}$ are infinitesimally $\Delta$-bounded, there is
  $\Sigma > 0$ such that $\sigma(H) \subset [-\Sigma, +\infty)$ for
  all $W$ with $\|W\|_{L^{2}_{\per}} \leq R$. Let $\cont$ be the
  contour given by Figure \ref{fig:cont}. This contour encloses the
  spectrum of $\sigma(H)$, avoids the poles of the Fermi-Dirac
  function at $\varepsilon_{F} + i\pi k_{B}T + 2i\pi k_{B}T \Z$, and is
  asymptotic to $\Im (z) = a \pm b \Re (z)$ for large $z \in \cont$,
  for some $a \in \R, b > 0$.
  The function $f_{\varepsilon_{F}}$ is therefore analytic inside
  $\cont$, decays exponentially when $|z| \to \infty$, and we have
  \begin{align*}
    f_{\varepsilon_{F}}(H) = \frac 1 {2\pi i}\int_{\cont} f_{\varepsilon_{F}}(z)(z-H)^{-1} dz
  \end{align*}
  as bounded operators\footnote{Note that if other occupation
    functions are used, the contour may need to be modified. For
    instance, Gaussian smearing \cite{cances2018numerical} decays
    exponentially only if $b < 1$. Our technique is less
    general than that of \cite{nier1993variational} based on the
    Helffer-Sj\"ostrand formula, which does not require any
    analyticity in $f_{\varepsilon_{F}}$.}. Let $R_{z} = (z-H)^{-1}$. Because
  $d(z,\sigma(H))$ increases at the same rate as $\Re(z)$, it follows
  from Lemma \ref{lemma:res-inv-lap} that
  $B_{z} = R_{z}(1-\Delta)$ is bounded in operator norm,
  independently of $z \in \cont$ (note that this would not be true for
  a rectangular contour). From the Kato-Seiler-Simon equality, there
  is therefore $C > 0$ such that, for all $z \in \cont, \delta W \in L^{2}_{\per}$,
  \begin{align*}
    \|R_{z} \delta W\|_{\Schatten_{2,\per}} \leq \|B_{z}\| \|(1-\Delta)^{-1} \delta W\|_{\Schatten_{2,\per}} \leq C \|\delta W\|_{L^{2}_{\per}}.
  \end{align*}
  Therefore, for $\|\delta W\|_{L^{2}_{\per}} \leq \frac 1 {2C}$, for all
  $z \in \cont$, $(z-H - \delta W)$ is invertible, and
  \begin{align*}
    (z-H - \delta W)^{-1} = (1-R_{z} \delta W)^{-1}R_{z} = \sum_{n \geq 0} (R_{z}\delta W)^{n} R_{z}.
  \end{align*}
  
  We can expand $f_{\varepsilon_{F}}(H + \delta W)$ as bounded operators:
  \begin{align*}
    f_{\varepsilon_{F}}(H + \delta W) - f_{\varepsilon_{F}}(H) &= \frac 1 {2\pi i}\int_{\cont} f_{\varepsilon_{F}}(z)\sum_{n \geq 1} (R_{z} \delta W)^{n} R_{z}dz.
  \end{align*}
  For all $n \geq 1$, we have that
  \begin{align*}
    \left\|\left(R_{z}  \delta W\right)^{n} R_{z}(1-\Delta)\right\|_{\Schatten_{2,\per}} &\leq \|B_{z}\| \|R_{z} \delta W\|_{\Schatten_{2,\per}}^{n} \leq C^{n} \|\delta W\|_{L^{2}_{\per}}^{n}.
  \end{align*}
  It follows from the decay properties of $f_{\varepsilon_{F}}$ on
  $\cont$ that
  \begin{align*}
    \int_{\cont} \sum_{n \geq 1} f_{\varepsilon_{F}}(z) \left\|\left(R_{z}  \delta W\right)^{n} R_{z}(1-\Delta)\right\|_{\Schatten_{2,\per}} < \infty,
  \end{align*}
  and therefore that $F_{\varepsilon_{F}}$ is analytic at $W$.
  
  \textbf{Step 3: $F_{\varepsilon_{F}}'$ is self-adjoint and
    non-positive.}
  From the previous computations, we have
  \begin{align*}
    F_{\varepsilon_{F}}'(W) \cdot \delta W = \frac 1 {2\pi i}\int_{\cont} f_{\varepsilon_{F}}(z)(R_{z} \delta W R_{z})(x,x) dz.
  \end{align*}
  with $(R_{z} \delta W R_{z})(x,x) \in L^{2}_{\per}$ uniformly in $z \in \cont$.

  For all $z \in \cont$, $R_{z} \delta W R_{z}$ is periodic with
  fibers $(z-H_{k})^{-1} \delta W (z-H_{k})^{-1}$. Inserting the spectral
  $(z-H_{k})^{-1} = \sum_{n \in \N} (z-\varepsilon_{nk})^{-1}|u_{nk}\rangle\langle u_{nk}|$, we get, for all $\delta W_{1},
  \delta W_{2} \in L^{2}_{\per}$,
  \begin{align*}
    &\langle \delta W_{1}, F_{\varepsilon_{F}}'(W) \cdot \delta W_{2}\rangle\\
    &=\frac 1 {2\pi i}\int_{\cont} f_{\varepsilon_{F}}(z)\fint_{\BZ} \sum_{n,m \in \N} \frac{1}{(z-\varepsilon_{nk})(z-\varepsilon_{mk})} \Tr\Big( \overline{\delta W_{1}}|u_{nk}\rangle\langle u_{nk}, \delta W_{2} u_{mk}\rangle\langle  u_{mk}| \Big)dk dz.\\
    &=\frac 1 {2\pi i}\int_{\cont} f_{\varepsilon_{F}}(z)\fint_{\BZ} \sum_{n,m \in \N} \frac{1}{(z-\varepsilon_{nk})(z-\varepsilon_{mk})} \langle  \delta W_{1} u_{mk},  u_{nk} \rangle\langle  u_{nk}, \delta W_{2} u_{mk} \rangle   dk dz.
  \end{align*}
  The absolute convergence of this sum in $L^{2}_{\per}$ follows from
  the estimates above. For completeness, we give a more direct proof.
  Let $e_{K}(x) = \frac{e^{iKx}}{\sqrt {|\Gamma|}}$. The
  $(e_{K})_{K \in \cR^{*}}$ form a Hilbert basis of $L^{2}_{\per}$,
  and we have the property
  $\langle e_{K}, \delta W e_{K'} \rangle = \frac{1}{\sqrt
    {|\Gamma|}}\langle e_{K-K'}, {\delta W} \rangle$ for all
  $\delta W \in L^{2}, K,K' \in \cR^{*}$. It follows that
  \begin{align*}
    \sum_{n,m \in \N} \frac{|\langle  u_{nk}, \delta W u_{mk} \rangle|^{2}}{|z-\varepsilon_{nk}|^{2}} &= \sum_{n,m \in \N, K,K' \in \cR^{*}} \frac{|\langle  u_{nk}, e_{K}\rangle\langle e_{K},\delta W e_{K'}\rangle\langle e_{K'},u_{mk} \rangle|^{2}}{|z-\varepsilon_{nk}|^{2}}\\
    &= \frac 1 {|\Gamma|}\sum_{n \in \N, K,K' \in \cR^{*}} \frac{|\langle  u_{nk}, e_{K}\rangle\langle e_{K-K'},\delta W\rangle|^{2}}{|z-\varepsilon_{nk}|^{2}}\\
    &= \frac {\|\delta W\|_{L^{2}_{\per}}^{2}} {|\Gamma|}\sum_{n \in \N, K \in \cR^{*}} \frac{|\langle  u_{nk}, e_{K}\rangle|^{2}}{|z-\varepsilon_{nk}|^{2}}\\
    &= \frac{\|\delta W\|_{L^{2}_{\per}}^{2}}{|\Gamma|} \sum_{n \in \N} \frac{1}{|z-\varepsilon_{nk}|^{2}}
  \end{align*}
  which is bounded uniformly in $k \in \BZ, z \in \cont$.
  The result follows by a Cauchy-Schwarz inequality. Performing the contour integration, we obtain the
  following sum-over-states formula
  \begin{align*}
      \langle \delta W_{1}, F_{\varepsilon_{F}}'(W) \cdot \delta W_{2}\rangle &= \fint_{k \in \BZ}\sum_{n, m \in \N} \frac{f_{nk} - f_{mk}}{\varepsilon_{nk} - \varepsilon_{mk}} \langle  \delta W_{1} u_{mk},  u_{nk} \rangle\langle  u_{nk}, \delta W_{2} u_{mk} \rangle dk
  \end{align*}
  where $f_{nk} = f_{\varepsilon_{F}}(\varepsilon_{nk})$, and with the
  convention that
  \begin{align*}
    \frac{f_{\varepsilon_{F}}(\varepsilon)-f_{\varepsilon_{F}}(\varepsilon)}{\varepsilon-\varepsilon} = f_{\varepsilon_{F}}'(\varepsilon)
  \end{align*}
  arising from the double pole $(z-\varepsilon_{mk})^{-2}$ when
  $\varepsilon_{nk} = \varepsilon_{mk}$. $F_{\varepsilon_{F}}'(W)$ is
  therefore self-adjoint and, since
\begin{align*}
  \langle \delta W, F_{\varepsilon_{F}}'(W) \cdot \delta W\rangle &= \fint_{k \in \BZ}\sum_{n, m \in \N} \frac{f_{nk} - f_{mk}}{\varepsilon_{nk} - \varepsilon_{mk}} \left| \langle  u_{nk}, \delta W u_{mk} \rangle  \right|^{2} dk
\end{align*}
and $f_{\varepsilon_{F}}$ is decreasing, it follows that $F_{\varepsilon_{F}}'(W)$ is
non-positive.

\textbf{Step 4: $F_{\varepsilon_{F}}'(W) + \beta \Delta$ is
  negative.} Assume that $F_{\varepsilon_{F}}'(W) + \beta\Delta$
is not negative. This means that there exists a sequence $W_{n}$ of
potentials with $\|W_{n}\|_{L^{2}_{\per}} = 1$ such that
\begin{align*}
  \langle W_{n}, F_{\varepsilon_{F}}'(W) \cdot W_{n}\rangle - \beta \int_{\Gamma}|\nabla W_{n}|^{2}  \to 0.
\end{align*}
Let $e \in L^{2}_{\per}$ be the constant
function: $e(x) = 1$. We have
$\sum_{K \neq 0} |K|^{2} |c_{K}(W_{n})|^{2} \to 0$, so that
$P_{e}^{\perp} W_{n}\to 0$ in $H^{1}$. Up to a subsequence, we can
assume that $\langle e, W_{n} \rangle \to c$, with $c \in \C$. It
follows from $W_{n} = P_{e} W_{n} + P_{e^{\perp}} W_{n}$ that
$W_{n} \to c e$ in $L^{2}$, and that $|c| = \frac 1 {\sqrt{|\Gamma|}}$. Then,
\begin{align*}
  \langle W_{n}, F_{\varepsilon_{F}}'(W) \cdot W_{n}\rangle  &\leq \fint_{k \in \BZ} f_{\varepsilon_{F}}'(\varepsilon_{1k}) |\langle  u_{1k}, W_{n} u_{1k} \rangle|^{2} dk \\
  &\to \frac 1 {\sqrt{|\Gamma|}}\fint_{k \in \BZ}f_{\varepsilon_{F}}'(\varepsilon_{1k})dk < 0.
\end{align*}
where we have used that $\||u_{1k}|^{2}\|_{L^{2}_{\per}} \leq C
\|u_{1k}\|_{H^{1}_{\per}}^{2}$ is bounded uniformly in $k$.
\end{proof}

\subsection{Self-consistent Fermi level}
We now solve the equation $\int_{\Gamma} F_{\varepsilon_{F}}(W) =
N_{\rm el}$ for $\varepsilon_{F}$.
\begin{lemma}
  \label{lem:fermi_level}
  For all $W \in L^{2}_{\per}$ and $N_{\rm el} > 0$, the equation $\int_{\Gamma}
  F_{\varepsilon_{F}}(W) = N_{\rm el}$ has a unique
  solution $\varepsilon_{F}(W)$. The map
  \begin{align*}
    F(W) &= F_{\varepsilon_{F}(W)}(W)
  \end{align*}
  is analytic from $L^{2}_{\per}$ to $L^{2}_{\per}$. Its
  differential $F'(W)$ is self-adjoint and non-positive, and satisfies
  $F'(W)\cdot e = 0$ where $e$ is the constant function.
\end{lemma}
\begin{proof}
  Let
  \begin{align*}
    \cN(\varepsilon_{F},W) = \langle  e, F_{\varepsilon_{F}}(W) \rangle = \fint_{\BZ} \sum_{n \in \N} f_{\varepsilon_{F}}(\varepsilon_{nk}) dk
  \end{align*}
  be the total number of electrons with Fermi level $\varepsilon_{F}$.
  Then by the previous lemma $\cN$ is analytic on
  $\R \times L^{2}_{\per}$, with
  \begin{align*}
    \frac {\partial \cN} {\partial \varepsilon_{F}}(\varepsilon_{F},W) &= - \langle e, F_{\varepsilon_{F}}'(W) \cdot
  e\rangle = - \fint_{\BZ} \sum_{n \in \N} f'_{\varepsilon_{F}}(\varepsilon_{nk}) dk> 0\\
    \frac {\partial \cN} {\partial W}(\varepsilon_{F},W) \cdot \delta W &= \langle  e, F_{\varepsilon_{F}}'(W) \cdot \delta W \rangle = \langle  F_{\varepsilon_{F}}'(W) \cdot e, \delta W \rangle
  \end{align*}
  For all $W \in L^{2}_{\per}$,
  $\cN(\cdot, W)$ has limit $0$ at $-\infty$ and $+\infty$ at
  $+\infty$, so that there is a unique solution $\varepsilon_{F}(W)$
  of $\cN(\varepsilon_{F},W) = N_{\rm el}$. From the implicit function
  theorem, we get that $\varepsilon_{F}(W)$ is analytic on
  $L^{2}_{\per}$ and
  \begin{align*}
    \varepsilon_{F}'(W) \cdot \delta W = \frac 1 {\langle e, F_{\varepsilon_{F}}'(W) \cdot
  e\rangle} \langle  F_{\varepsilon_{F}}'(W) \cdot e, \delta W\rangle
\end{align*}
and therefore
  \begin{align*}
    F'(W) \cdot \delta W&= F'_{\varepsilon_{F}}(W) \cdot \delta W - \frac{1}{\langle e, F_{\varepsilon_{F}}'(W)\cdot e\rangle}  \langle  F_{\varepsilon_{F}}'(W) \cdot e, \delta W \rangle \; F_{\varepsilon_{F}}'(W)\cdot e
  \end{align*}
  
  In particular, $F'(W)\cdot e = 0$ and $F'(W)$ is self-adjoint.

  The expression above is of the form
  $F'(W) = A - \frac 1 {\langle e, A e \rangle} |Ae\rangle\langle
  Ae|$, with $A = F_{\varepsilon_{F}}'(W)$ bounded, self-adjoint and
  non-positive on a Hilbert space $\cH$ and $e \in \cH$. In
  particular, $F'(W)$ is self-adjoint, $F'(W) \cdot e = 0$, and we
  compute, for all $x \in \cH$,
  \begin{align*}
    \left\langle x, \left( A - \frac 1 {\langle e, A e
          \rangle} |Ae\rangle\langle Ae|\right) x \right\rangle &= -\left(\||A|^{1/2}x\|^{2} - \frac{|\langle |A|^{1/2}x, |A|^{1/2} e \rangle|^{2}}{\||A|^{1/2}e\|^{2}}\right) \leq 0
\end{align*}
by the Cauchy-Schwartz inequality, from where it follows that $F'(W)$
is non-positive.

\end{proof}

We now look for solutions $W \in L^{2}_{\per}$ of the
equation
\begin{align*}
  W = W_{\nucl} + \vcper F(W).
\end{align*}
We are ready for the
\begin{proof}[Proof of Theorem \ref{thm:main_periodic}.]
  Set
  \begin{align*}
    M_{\alpha}(W,W_{\nucl}) = W + \alpha(W_{\nucl} + \vcper F(W) - W)
  \end{align*}
  In particular, $M_{\alpha}$ is analytic from $L^{2}_{\per}$ to
itself. From Theorem \ref{thm:abstract_FP}, we only need to check that
the spectral radius of the operator
  \begin{align*}
    J_{\alpha} = 1 + \alpha(\vcper F'(W^{*}) - 1)
  \end{align*}
  is smaller than $1$ for $\alpha$ small enough. This is ensured by
  the fact that $\vcper$ and $F'(W^{*})$ are bounded operators on
  $L^{2}_{\per}$ and $\vcper$ is non-negative, so that, from Lemma \ref{lem:sigma_AB},
  \begin{align*}
    \sigma(\vcper F'(W^{*})) \setminus \{0\} = \sigma(\sqrt \vcper F'(W^{*}) \sqrt \vcper) \setminus \{0\}
  \end{align*}
  This last operator is a non-positive self-adjoint bounded operator
  on $L^{2}_{\per}$, hence the result.
\end{proof}

\section{The defect problem}
\label{sec:defect}
In this section we fix $W_{\per} \in L^{2}_{\per}$ and
$\varepsilon_{F} \in \R$. Let
\begin{align*}
  H_{\per} = -\Delta + W_{\per}
\end{align*}
be the background periodic Hamiltonian.

We first investigate the renormalized potential-to-density mapping.
\begin{lemma}
  There is a neighborhood $\mathcal V$ of $0$ in $L^{2}$ in which the
  map
\begin{align*}
  G(V) = (f_{\varepsilon_{F}}(H_{\per} + V) - f_{\varepsilon_{F}}(H_{\per}))(x,x)
\end{align*}
is analytic from $L^{2}_{N}$ to $L^{2}_{N}$ for all $N \geq 0$.

Let $\chi_{0} = G'(0)$. Then $V \mapsto G(V) - \chi_{0} V$ maps
$\mathcal V \cap L^{2}_{N}$ to $L^{2}_{2N}$ for all $N \geq 0$.
\end{lemma}
\begin{proof}
  \textbf{Step 1: the case $N = 0$.} The proof of this step is
  similar to that of Lemma \ref{lem:F_per}. We take a contour $\cont$
  enclosing the spectrum of $H_{\per}$ with the same shape as in
  Figure \ref{fig:cont}, which encloses the spectrum of
  $H_{\per} + V$ for $\|V\|_{L^{2}}$ small because $L^{2}$ potentials
  are infinitesimally $\Delta$-bounded. From Lemma
  \ref{lemma:res-inv-lap}, there is $C > 0$ such that, for all
  $z \in \cont$ and $V \in L^{2}$,
  \begin{align*}
    \|R_{z}V\|_{\Schatten_{2}} = \|R_{z}(1-\Delta)(1-\Delta)^{-1}V\|_{\Schatten_{2}} \leq C \|V\|_{L^{2}}.
  \end{align*}
  with $R_{z} = (z-H_{\per})^{-1}$. It follows that, for
  $\|V\|_{L^{2}} \leq \frac 1 {2C}$, $(1-R_{z} V)$ is invertible, and we
  have
  \begin{align}
    f_{\varepsilon_{F}}(H_{\per} + V) -
    f_{\varepsilon_{F}}(H_{\per}) &= \frac 1 {2\pi i}\int_{\cont} f_{\varepsilon_{F}}(z)\left( \left(z- (H_{\per} + V)\right)^{-1} - (z-H_{\per})^{-1} \right)dz \notag\\
    &= \frac 1 {2\pi i} \int_{\cont}f_{\varepsilon_{F}}(z) \sum_{n\geq 1} (R_{z} V)^{n} R_{z} dz
    \label{eq:exp_V_def}
  \end{align}
  as bounded operators. From the estimate
  \begin{align*}
    \|(R_{z} V)^{n-1} R_{z}V R_{z}(1-\Delta)\|_{\Schatten_{2}} \leq \|R_{z}(1-\Delta)\| \|R_{z} V\|_{\Schatten_{2}}^{n} \leq C^{n} \|V\|_{L^{2}}^{n}
  \end{align*}
  with $C$ uniform in $z \in \cont$ and the decay properties of
  $f_{\varepsilon_{F}}$, it follows that the expansion
  \eqref{eq:exp_V_def} converges absolutely. Therefore,
  $f_{\varepsilon_{F}}(H_{\per} + V) - f_{\varepsilon_{F}}(H_{\per})$
  can be associated a density $G(V) \in L^{2}$, and $G$ is analytic in a
  neighborhood of $0$.
  
  \textbf{Step 2: Bloch structure of the expansion of the density at
    all orders.}   We first note that
  \begin{align*}
    H_{k+q} - H_{k} = 2 (-i\nabla + k) \cdot q + |q|^{2}.
  \end{align*}
  The bounded operator $R_{z} = (z-H)^{-1}$ on $L^{2}$ is periodic with fibers
  $R_{z,k} = (z-H_{k})^{-1}$. Since $R_{z,k}(1-\Delta)$ is bounded
  uniformly in $z \in \cont$ and $k \in \BZ$,
  \begin{align*}
    \|R_{z,k}(2 (-i\nabla + k) \cdot q + |q|^{2})\| = \|R_{z,k}(1-\Delta)(1-\Delta)^{-1}(2 (-i\nabla + k) \cdot q + |q|^{2})\| \leq C (|q|+|q|^{2}).
  \end{align*}
  For $q$ small enough, we then have
  \begin{align*}
    R_{z,k+q} = \sum_{n \geq 0} \Big(R_{z,k}(2(-i\nabla + k) \cdot q+|q|^{2})\Big)^{n} R_{z,k}
  \end{align*}
  and, from the previous estimate, $R_{z,k}$ is analytic in the
  $\|\cdot (1-\Delta)\|_{\Schatten_{2,\per}}$ topology, uniformly in
  $z \in \cont$ and $k \in \BZ$.

  For $z \in \cont$, let
  \begin{align*}
    D^{(n)}_{z}(V) = ((R_{z} V)^{n}R_{z})(x,x).
  \end{align*}
  
  We first consider the first-order term $D^{(1)}_{z}$. Let $V = \int_{\BZ} e^{iqx} V_{q}(x)dq \in L^{2}$. Elementary
  computations show that if $A$ is a periodic operator with fibers
  $A_{k}$, then $e^{-iqx} A e^{iqx}$ is a periodic operator with
  fibers $A_{k+q}$, and that $e^{iqx} A$ has density
  \begin{align*}
    (e^{iqx} A)(x,x) = e^{iqx} A(x,x) = e^{iqx} \fint_{\BZ} A_{k}(x,x) dk.
  \end{align*}
  Therefore,
  \begin{align*}
    D_{z}^{(1)}(V) &= \int_{\BZ} (e^{iqx} e^{-iqx} R_{z} e^{iqx} V_q R_{z})(x,x) dq= \int_{\BZ} e^{iqx} \fint_{\BZ}(R_{z,k+q} V_q R_{z,k})(x,x) dk\, dq\\
    D_{z}^{(1)}(V)_{q} &= \fint_{\BZ}(R_{z,k+q} V_q R_{z,k})(x,x) dk.
  \end{align*}
  Similarly, in the general case,
  \begin{align}
    \label{eq:def_Dn} D_{z}^{(n)}(V)_{q} = \int_{q_{1},\dots,q_{n-1} \in \BZ} \fint_{k \in \BZ} (R_{z,k+q} V_{q-q_{1}} R_{z,k+q_{1}}V_{q_{1}-q_{2}} \cdots R_{z,k})(x,x) dk\, dq_{1}\cdots dq_{n-1} 
  \end{align}

  \textbf{Step 3: the case $N > 0$.}
  Since for $i = 1,2,3$
  \begin{align}
    \label{eq:diff_Dn}
  (x_{i} D^{(n)}_{z}(V))_{q} &= i \partial_{q_{i}} D^{(n)}_{z}(V)_{q}
\end{align}
and $R_{z,k}$ is analytic for the
$\|\cdot (1-\Delta)\|_{\Schatten_{2,\per}}$ topology, uniformly in
$z \in \cont, k \in \BZ$, the repeated application of
\eqref{eq:diff_Dn} to \eqref{eq:def_Dn} yields a bound of the form 
\begin{align*}
  \|D^{(n)}_{z}(V)\|_{L^{2}_{N}} \leq C_{N} \|V\|_{L^{2}_{N}} C^{n-1}\|V\|_{L^{2}}^{n-1}
\end{align*}
for all $N \geq 0$, with $C$ independent on $N$, this bound being uniform in
$z \in \cont$.

It follows that, for $\|V\|_{L^{2}} \leq \frac 1 {2C}$, for all $N$, we
have the absolutely convergent expansion
  \begin{align*}
    G(V) &= \frac 1 {2\pi i} \int_{\cont}f_{\varepsilon_{F}}(z) \sum_{n \geq 1} D^{(n)}_{z}(V) dz
  \end{align*}
  in $L^{2}_{N}$.

  \textbf{Step 4: $G(V) - \chi_{0} V : L^{2}_{N} \mapsto L^{2}_{2N}$.}
  We have
  \begin{align*}
    G(V) - \chi_{0} V &= \frac 1 {2\pi i} \int_{\cont}f_{\varepsilon_{F}}(z) \sum_{n \geq 2} D^{(n)}_{z}(V) dz.
  \end{align*}
  Consider terms of the form
  \begin{align}
    \label{eq:form_terms}
    I(q)=&\int_{q_{1},\dots,q_{n-1} \in \BZ} \fint_{k \in \BZ} (R_{z,k+q}^{(1)} V^{(1)}_{q-q_{1}} R_{z,k+q_{1}}^{(2)}V^{(2)}_{q_{1}-q_{2}} R_{z,k+q_{2}} V_{q_{2}-q_{3}}\cdots R_{z,k})(x,x)  dk dq_{1}\cdots dq_{n-1}
  \end{align}
  where $R^{(1)}_{z,k},R^{(2)}_{z,k}$ are $R_{z,k}$ or their
  derivatives, and $V^{(1)}_{q},V^{(2)}_{q}$ are $V_{q}$ or its
  derivatives. Performing the change of variable $q_{1}' = q-q_{1}$, we
  obtain
  \begin{align*}
    I(q)&=- \int_{q_{1}' \in q-\BZ} \int_{q_{2},\dots,q_{n-1} \in \BZ} \fint_{k \in \BZ} (R_{z,k+q}^{(1)} V^{(1)}_{q_{1}'} R_{z,k+q-q_{1}'}^{(2)}V^{(2)}_{q-q_{1}'-q_{2}} R_{z,k+q_{2}} V_{q_{2}-q_{3}}\cdots R_{z,k})(x,x)  dk dq_{1}'\cdots dq_{n-1}\\
    &=- \int_{q_{1}',\dots,q_{n-1} \in \BZ} \fint_{k \in \BZ} (R_{z,k+q}^{(1)} V^{(1)}_{q_{1}'} R_{z,k+q-q_{1}'}^{(2)}V^{(2)}_{q-q_{1}'-q_{2}} R_{z,k+q_{2}} V_{q_{2}-q_{3}}\cdots R_{z,k})(x,x)  dk dq_{1}'\cdots dq_{n-1},
  \end{align*}
  where the quasi-periodicity of the Bloch transform implies that
  $q_{1}' \mapsto V^{(1)}_{q_{1}'}
  R_{z,k+q-q_{1}'}^{(2)}V^{(2)}_{q-q_{1}'-q_{2}}$ is
  $\cR^{*}$-periodic, and therefore that we can integrate $q_{1}'$
  over $\BZ$ rather than $q-\BZ$. This shows that we can transfer the
  $q$ dependence from $V^{(1)}$ to $V^{(2)}$ in the convolution-like
  terms of the form \eqref{eq:form_terms}.
  

  Let $p_{2N}$ be a polynomial of degree $2N$. Applying
  \eqref{eq:diff_Dn} successively to $p_{2N}(x) D_{z}^{(n)}(V)$ and
  using the above procedure to the divide the $2N$ derivatives between
  $V^{(1)}_{q}$ and $V^{(2)}_{q}$, we obtain that
  $(p_{2N}(x) D_{z}^{(n)}(V))_{q}$ contains terms of the form
  \eqref{eq:form_terms} with $V^{(1)}_{q}$ and $V^{(2)}_{q}$ being
  derivatives of $V_{q}$ of order at most $N$. Using the analyticity
  of $k \mapsto R_{z,k}$, we obtain a bound of the form
  \begin{align*}
    \|D^{(n)}_{z}(V)\|_{L^{2}_{2N}} \leq C_{N} \|V\|_{L^{2}_{N}}^{2} C^{n-2} \|V\|_{L^{2}}^{n-2}
  \end{align*}
  where $C$ is independent of $N$, the bound being uniform in
  $z \in \cont$. The result follows.

\end{proof}

Recall that the operator $\vc$ is given by the convolution
\begin{align*}
  (\vc \rho)(x) = \frac 1 {4\pi}\int_{\R^{3}} \frac{\rho(y)}{|x-y|} dy.
\end{align*}
In Fourier space, this is a multiplication by $\frac 1 {|q|^{2}}$.
This is an unbounded non-negative self-adjoint operator on $L^{2}$. We denote its
formal inverse by $\vc^{-1} = -\Delta$, also an unbounded non-negative
self-adjoint operator on $L^{2}$. $\vc^{-1}$ does not have a spectral
gap at zero, but $-\chi_{0} + \vc^{-1}$ does:
\begin{lemma}
  \label{lem:chi_0_def}
  Let $N \geq 0$. The operator $-\chi_{0} + \vc^{-1}$
  is self-adjoint and positive on $L^{2}$, and its inverse is bounded from
  $H^{-2}_{N}$ to $L^{2}_{N}$.
  The operator $\varepsilon = 1 - \vc \chi_{0}$ is invertible in
  $L^{2}_{N}$, with bounded inverse
  \begin{align*}
    \varepsilon^{-1} = (1-\vc \chi_{0})^{-1} = \left(-\chi_{0} + \vc^{-1}\right)^{-1} \vc^{-1}.
  \end{align*}
  The operator $\kerker \varepsilon$ is therefore bounded and
  invertible on $L^{2}_{N}$.
\end{lemma}
\begin{proof}
  We have, for $V \in L^{2}$ with Bloch transform $V_{q}$
  \begin{align*}
    \chi_{0} V &= \frac 1 {2\pi i} \int_{\cont}f_{\varepsilon_{F}}(z) \int_{\BZ} \fint_{\BZ}(R_{z,k+q} V_q R_{z,k})(x,x) dk  dq dz
  \end{align*}
  and therefore $\chi_{0}$ is fibered, with fibers
  \begin{align}
    \label{eq:chi0q}
    \chi_{0,q} W &= \frac 1 {2\pi i} \int_{\cont}f_{\varepsilon_{F}}(z)  \fint_{\BZ}(R_{z,k+q}W  R_{z,k})(x,x) dk  dz
  \end{align}
  for $W \in L^{2}_{\per}$. As in Lemma \ref{lem:F_per}, inserting the
  decomposition
  $R_{z,k} = \sum_{n \in \N} (z-\varepsilon_{nk})^{-1}
  |u_{nk}\rangle\langle u_{nk}|$, we obtain the sum-over-states
  formula
  \begin{align*}
    \langle  W_{1}, \chi_{0,q} W_{2}\rangle &= \fint_{\BZ}  \sum_{n,m \geq 0} \frac{f_{n,k+q}-f_{m,k}}{\varepsilon_{n,k+q}-\varepsilon_{m,k}} \langle W_{1} u_{m,k}, u_{nk+q}\rangle\langle u_{nk+q},W_{2} u_{m,k}\rangle dk
  \end{align*}
  converging absolutely, from where it follows that $\chi_{0,q}$ is
  self-adjoint and non-positive on $L^{2}$ for all $q$, and therefore
  that $\chi_{0}$ is too.

  It follows from the regularity of $R_{z,k}$ and \eqref{eq:chi0q}
  that $\chi_{0,k}$ is analytic as bounded operators in
  $L^{2}_{\per}$, with $\chi_{0,0} = F_{\varepsilon_{F}}'(W_{\per})$.
  The operator $\vc^{-1} = -\Delta$ has fibers
  $\vck^{-1} = (-i\nabla+k)^{2}$ positive except at $k = 0$. Using
  Lemma \ref{lem:F_per} with $\beta = 1/2$, it follows that
  \begin{align*}
    -\chi_{0,k} + \frac 1 2 \vck^{-1} =
    -(\chi_{0,k}-F_{\varepsilon_{F}}'(W_{\per})) -
    F_{\varepsilon_{F}}'(W_{\per}) + \frac 1 2 \vck^{-1}
  \end{align*}
  is bounded
  away from zero for $k$ small, and therefore for all $k$. Therefore,
  there is $c > 0$ such that
  \begin{align*}
    -\chi_{0} + \vc^{-1} = -\chi_{0} + \frac 1 2 \vc^{-1} + \frac 1 2 \vc^{-1}\geq c(1-\Delta)
  \end{align*}
  as quadratic forms, from where it follows that
  $(-\chi_{0} + \vc^{-1})^{-1} \leq \frac {1} c(1-\Delta)^{-1}$ as
  quadratic forms and then that, for all $V \in L^{2}$,
  $\|(-\chi_{0} + \vc^{-1})^{-1}V\|_{L^{2}} \leq \frac 1 {c}
  \|(1-\Delta)^{-1}V\|_{L^{2}}$. The operator
  $(-\chi_{0} + \vc^{-1})^{-1}$ is therefore bounded from $H^{-2}$ to
  $L^{2}$.

  Its fibers are $(\chi_{0,k} + \vck^{-1})^{-1}$ and, for $q$ small
  enough,
  \begin{align*}
    (\chi_{0,k+q} + \vck^{-1})^{-1} =  \sum_{n \geq 0} \Big(  (\chi_{0,k} + \vck^{-1})^{-1} (\chi_{0,k+q} - \chi_{0,k})^{n} \Big) (\chi_{0,k} + \vck^{-1})^{-1}
  \end{align*}
  which shows that the family $(\chi_{0} + \vc^{-1})^{-1}_{k}$ is
  analytic on $\BZ$ as operators from $H^{-2}_{\per}$ to
  $L^{2}_{\per}$, and therefore that $(\chi_{0} + \vc^{-1})^{-1}$ is
  bounded from $H^{-2}_{N}$ to $L^{2}_{N}$. It then follows that
  $\varepsilon$ is invertible on $L^{2}_{N}$, with inverse
  \begin{align*}
    \varepsilon^{-1} = (1-\vc \chi_{0})^{-1} = \left(-\chi_{0} + \vc^{-1}\right)^{-1} \vc^{-1}.
  \end{align*}

  Finally, we have
  \begin{align*}
    (\kerker \varepsilon)^{-1} = (-\chi_{0} + \vc^{-1})^{-1} \vc^{-1} \kerker^{-1} = (-\chi_{0} + \vc^{-1})^{-1} (1-\Delta),
  \end{align*}
  hence the result.
\end{proof}

We are now ready for the
\begin{proof}[Proof of Theorem \ref{thm:existence_defect}.]
  We proceed as in Theorem \ref{thm:main_periodic}, and apply Theorem
  \ref{thm:abstract_FP} to
  \begin{align*}
    M(V,V_{\defect}) = V + \alpha \kerker (V_{\defect} + \vc G(V) - V),
  \end{align*}
  analytic in a neighborhood of $0$ in
  $L^{2} \times \vc H^{-2}$ to $L^{2}$, with Jacobian at
  $(0,0)$
  \begin{align*}
    J_{\alpha} = 1 -\alpha \kerker + \alpha \kerker \vc \chi_{0} = 1 - \alpha \kerker \varepsilon.
  \end{align*}
  Since ${\kerker \vc} = (1-\Delta)^{-1}$ is bounded, self-adjoint and
  non-negative on $L^{2}$, we have
  \begin{align*}
    \sigma(\kerker \vc \chi_{0}) \setminus \{0\}= \sigma(\sqrt{\kerker \vc} \chi_{0} \sqrt{\kerker \vc}) \setminus \{0\}.
  \end{align*}
  It follows that 
  by taking $\alpha_{0}$ small enough, we can impose that
  $\sigma(J_{\alpha}) \subset (-1,1]$. Since from Lemma
  \ref{lem:chi_0_def} the operator $\kerker \varepsilon$ is invertible
  on $L^{2}$, we even have that
  $\sigma(J_{\alpha}) \subset (-1,1)$, hence the result.
\end{proof}

\begin{proof}[Proof of Theorem \ref{thm:decay}.]
  The proof is based on a bootstrap argument on the equation
  \begin{align}
    \label{eq:bootstrap}
    V = \varepsilon^{-1}(V_{\defect} + \vc(G(V) - \chi_{0} V))
  \end{align}
  satisfied by $V(V_{\defect})$.

  For the base case $N=1$, we prove that $V(V_{\defect}) \in L^{2}_{1}$ by applying
  Theorem \ref{thm:abstract_FP} to
  \begin{align*}
    M(V,V_{\defect}) = \varepsilon^{-1}(V_{\defect} + \vc(G(V) - \chi_{0} V)),
  \end{align*}
  an analytic map from $L^{2}_{1} \times \vc H^{-2}_{1}$ to
  $L^{2}_{1}$ with Jacobian $0$ at $(0,0)$. It follows from the
  uniqueness of $V(V_{\per})$ that $V(V_{\per}) \in L^{2}_{1}$.

  We then use the fact that $M(\cdot,V_{\defect})$ maps $L^{2}_{1}$ to
  $L^{2}_{2}$ to conclude from
  \eqref{eq:bootstrap} that $V(V_{\per}) \in L^{2}_{2}$. Repeating
  this argument, we obtain that $V(V_{\per}) \in L^{2}_{N}$.
\end{proof}

\section*{Acknowledgments}
Stimulating discussions with Eric Canc\`es, Thierry Deutsch and
Domenico Monaco are gratefully acknowledged.
\bibliographystyle{plain}
\bibliography{scf}
\end{document}